\documentclass[a4paper,12pt]{article}
\usepackage{test}

\usepackage{xcolor}
\definecolor{pantone}{RGB}{1, 33, 105}

\usepackage{setspace} 
\usepackage{hyperref}
\hypersetup{colorlinks,linkcolor=blue,citecolor=blue,urlcolor=pantone}
\usepackage{doi}

\usepackage[style=ext-authoryear-comp,
sorting=nyt, 
sortcites=false, 
dashed=false, 
maxcitenames=2, 
maxbibnames=99, 
uniquelist=false,
uniquename=false,
giveninits=true, 
natbib, 
date=year 
]{biblatex}

\AtBeginRefsection{\GenRefcontextData{sorting=ynt}}
\AtEveryCite{\localrefcontext[sorting=ynt]}

\DeclareFieldFormat{pages}{#1} 
\renewbibmacro{in:}{\ifentrytype{article}{}{\printtext{\bibstring{in}\intitlepunct}}} 

\DeclareFieldFormat[article,inbook,incollection,inproceedings,patent,thesis,unpublished]{titlecase:title}{\MakeSentenceCase*{#1}} 

\newbibmacro*{atoda:titlelink}[1]{%
  \ifhyperref
    {\iffieldundef{doi}
       {\iffieldundef{eprint}
          {\iffieldundef{url}
             {#1}
             {\href{\thefield{url}}{#1}}}
          {\href{https://arxiv.org/abs/\thefield{eprint}}{#1}}}
       {\href{https://doi.org/\thefield{doi}}{#1}}}
    {#1}}

\DeclareFieldFormat{title}{\usebibmacro{atoda:titlelink}{\mkbibemph{#1}}}
\DeclareFieldFormat
  [article,inbook,incollection,inproceedings,patent,thesis,unpublished]
  {title}{\usebibmacro{atoda:titlelink}{\mkbibquote{#1\isdot}}}

\renewbibmacro*{doi+eprint+url}{}

\usepackage[inline,shortlabels]{enumitem}
\setlist[enumerate,1]{label=(\roman*)}

\usepackage{booktabs}
\usepackage{tabularx}
\usepackage{caption}
\usepackage{pgfplots}
\pgfplotsset{compat=1.18}
\usepgfplotslibrary{groupplots}

\usepackage[margin = 1.25in]{geometry}
\newcommand{\cA}{\mathcal{A}}
\newcommand{\cD}{\mathcal{D}}
\newcommand{\cF}{\mathcal{F}}
\newcommand{\cH}{\mathcal{H}}
\newcommand{\cX}{\mathcal{X}}
\newcommand{\OS}{\mathcal{OS}}

\title{Conditional Impatience and Concavity of Consumption Functions}
\author{Alexis Akira Toda\thanks{Department of Economics, Emory University and Research Institute for Economics and Business Administration, Kobe University. Email: \href{mailto:alexis.akira.toda@emory.edu}{alexis.akira.toda@emory.edu}.}}

\begin{document}
\maketitle

\begin{abstract}
Concave consumption functions imply a marginal propensity to consume that falls with wealth. I characterize the utility functions that guarantee this property in finite-horizon optimal saving problems with stochastic discounting, returns, income, and borrowing limits. Under conditional impatience---the conditional expected discounted gross return does not exceed one---consumption functions are always concave if and only if inverse absolute prudence, $-u''/u'''$, is concave. When no conditional-impatience restriction is imposed, hyperbolic absolute risk aversion (HARA) is necessary and sufficient for uniform concavity. Thus conditional impatience permits declining marginal propensities to consume for a preference class strictly larger than HARA.

\medskip


\end{abstract}

\section{Introduction}

How strongly consumption responds to available resources is central to both individual saving behavior and aggregate fluctuations. Empirical estimates show that marginal propensities to consume vary systematically with household balance sheets and the size of income shocks \citep{FagerengHolmNatvik2021}. In heterogeneous-agent models, their distribution helps determine the aggregate consumption response to shocks \citep{KaplanViolante2022}; differences in marginal propensities to consume also shape the redistribution channel of monetary policy \citep{Auclert2019}. A household consumption function that is concave in wealth provides a structural source of such heterogeneity: the same increment to current resources produces a larger increase in consumption at low wealth than at high wealth. Consequently, the distribution of resources can affect aggregate spending at a fixed level of total resources, and transfers of equal size can induce heterogeneous consumption responses. This property formalizes an idea going back to \citet[p.~120]{Keynes1936}, who argued that the marginal propensity to consume tends to decline as employment and real income rise. Concavity also changes consumption dynamics relative to the linear benchmark and plays an important role in buffer-stock saving and wealth-dependent responses to risk \citep{Kimball1990,CarrollKimball1996,CarrollKimball2008,Toda2021JME}.

Existing theory identifies hyperbolic absolute risk aversion (HARA) utility as the benchmark for uniform consumption concavity: consumption functions are concave in every optimal saving problem within a specified class. \citet{CarrollKimball1996} show that HARA guarantees this property in a broad class of finite-horizon saving problems under uncertainty, and \citet*{MaStachurskiToda2020JET} extend the underlying concavity-preservation argument to broader recursive environments. Conversely, \citet{Toda2021JME} proves that constant relative risk aversion (CRRA) is necessary on the consumption domain $C=(0,\infty)$ when discounting, returns, and income are unrestricted and explains how translating the domain yields HARA more generally. Together, these results suggest that, without a restriction on saving incentives, optimal saving does not robustly produce a declining marginal propensity to consume outside HARA. Can an economically relevant restriction support a larger class of utility functions?

Impatience provides such a restriction. Allowing income to be stochastic, \citet[Corollary~1]{ChamberlainWilson2000} show that consumption diverges in an infinite-horizon saving problem with a constant discount factor and gross return satisfying $\beta R>1$. Using this result, \citet[proof of Theorem~8]{StachurskiToda2019JET} show that market clearing in a stationary Bewley--Huggett--Aiyagari economy with a finite-mean wealth distribution requires $\beta R\le 1$. When discount factors and returns are stochastic, the natural one-step analogue is $\E_t[\beta_{t+1}R_{t+1}]\le 1$ almost surely at every date, a condition I call \emph{conditional impatience}. It restricts only the conditional mean of the discounted gross return and therefore permits $\beta_{t+1}R_{t+1}>1$ in some future states. In this paper, I ask which utility functions guarantee concavity when saving incentives satisfy this restriction and provide an exact answer.

I study a finite-horizon problem with jointly stochastic discount factors, gross returns, non-financial income, and borrowing limits. Under standard monotonicity, concavity, prudence, and Inada conditions on utility, consumption is concave in wealth in every conditionally impatient problem if and only if inverse absolute prudence, $-u''/u'''$, is concave. For HARA utility, inverse absolute prudence is affine. Conditional impatience therefore enlarges the class of utility functions from HARA to all utilities with concave inverse absolute prudence. The characterization holds history by history with stochastic borrowing limits, arbitrary dependence among the primitive shocks, and shocks with possibly infinite support. The curvature condition can be verified directly from utility without solving the household problem. Example~\ref{exmp:nonHARA} illustrates the larger class with a non-HARA family that nests CRRA utility.

The unrestricted result shows exactly what conditional impatience changes. If discounted returns face no restriction, consumption is concave in every saving problem if and only if utility is HARA. Moreover, the failure outside HARA does not require a long horizon, uncertainty, or a borrowing constraint. I give a deterministic one-period problem with unrestricted borrowing in which the non-HARA utility with marginal utility $u'(x)=\log(1+1/x)$ has strictly concave inverse absolute prudence but generates a strictly \emph{convex} consumption function when $\beta R=2$. This example shows that concave inverse absolute prudence alone does not guarantee concave consumption when conditional impatience is dropped, even in a deterministic problem with unrestricted borrowing.

I prove both theorems by studying how the Euler equation aggregates marginal utility across future states. Under conditional impatience, the relevant aggregation weights sum to at most one. I show that this aggregation preserves concavity exactly when inverse absolute prudence is concave, and backward induction then yields sufficiency. For necessity, I embed arbitrary finite-dimensional aggregation problems in one-period saving problems with finite-support shocks. Unrestricted weights force inverse absolute prudence to be affine and recover HARA. The two characterization theorems therefore follow from the same aggregation argument.

I build most directly on \citet{CarrollKimball1996}, \citet{SuenConcave}, and \citet{Toda2021JME}. Relative to \citet{CarrollKimball1996}, I characterize the larger class of utility functions supported by conditional impatience. \citet{SuenConcave} proves sufficiency of concave inverse absolute prudence under $\beta R\le 1$ in a finite-horizon model with a constant discount factor and risk-free return, deterministic age-dependent borrowing limits, and independent permanent and transitory income shocks with compact support. I establish necessity and allow discounting, returns, income, and borrowing limits to be jointly stochastic. Relative to \citet{Toda2021JME}, I state the unrestricted HARA benchmark in the present setting and identify precisely how conditional impatience changes it. Related work includes \citet*{MaStachurskiToda2020JET} and \citet*{CarrollHolmKimball2021}. Section~\ref{subsec:literature} gives a detailed comparison with this literature.

To isolate the role of uncertainty, I also examine perfect foresight with unrestricted borrowing. Under the discounted-return restrictions considered in Section~\ref{sec:perfectForesight}, consumption curvature is characterized by the curvature of absolute risk tolerance, $-u'/u''$. The analysis explains the deterministic example, corrects a claim in the literature, and shows that the perfect-foresight characterizations do not extend to uncertainty even when the same restrictions on discounted returns hold almost surely.

Section~\ref{sec:model} formulates the optimal saving problem, states the two characterization theorems, presents the analytical example, and discusses the relation to existing results. Section~\ref{sec:perfectForesight} examines perfect foresight and compares its curvature characterizations with those under uncertainty. Section~\ref{sec:conclusion} concludes. The appendix develops the aggregation lemmas and proves the results.

\section{Optimal saving and concavity of consumption}\label{sec:model}

\subsection{Optimal saving problem}\label{subsec:problem}

I consider the following optimal saving problem:
\begin{subequations}\label{eq:os}
    \begin{align}
        &\maximize && \E_0\sum_{t=0}^T\left(\prod_{j=0}^t\beta_j\right)u(c_t) \label{eq:os_obj}\\
        &\st && (\forall t) c_t\in C, \label{eq:os_feasible}\\
        &&& (\forall t) c_t\le w_t+b_t, \label{eq:os_borrow}\\
        &&& (\forall t<T) w_{t+1}=R_{t+1}(w_t-c_t)+Y_{t+1}, \label{eq:os_budget}
    \end{align}
\end{subequations}
where initial wealth $w_0$ is given. The agent maximizes expected discounted utility \eqref{eq:os_obj}. Constraint \eqref{eq:os_feasible} restricts consumption $c_t$ to the domain $C$ of the utility function $u$, while \eqref{eq:os_borrow} limits it to current wealth $w_t$ plus the borrowing limit $b_t\ge 0$. Equation \eqref{eq:os_budget} gives next-period wealth as savings $w_t-c_t$ multiplied by the gross return $R_{t+1}$, plus non-financial income $Y_{t+1}$.

In the objective function \eqref{eq:os_obj}, $T\in \set{0,1,\dotsc}$ is the horizon of the agent, $\beta_0\equiv 1$, $\beta_t>0$ is the discount factor between time $t-1$ and time $t$, and $u$ is the period utility function with consumption domain $C\subset\R$. I maintain the following assumption throughout.

\begin{asmp}[Utility]\label{asmp:utility}
The set $C$ is a nonempty open interval, and $u:C\to\R$ is three times continuously differentiable with $u'>0$, $u''<0$, $u'''>0$. Moreover, $u$ satisfies the Inada conditions
\begin{equation}
    \lim_{x\downarrow \inf C}u'(x)=\infty \quad \text{and} \quad \lim_{x\uparrow \sup C}u'(x)=0. \label{eq:Inada}
\end{equation}
\end{asmp}

Recall that the utility function $u$ exhibits \emph{hyperbolic absolute risk aversion} (HARA) if absolute risk aversion is hyperbolic, that is,
\begin{equation}
    -\frac{u''(x)}{u'(x)}=\frac{1}{ax+b} \label{eq:ARA}
\end{equation}
for some constants $a,b$, where the domain of $u$ is
\begin{equation*}
	C\coloneqq \set{x\in \R:ax+b>0}.
\end{equation*}
Integrating the differential equation \eqref{eq:ARA} twice shows that HARA utilities are completely characterized as
\begin{equation}
    u(x)=\begin{cases*}
        \frac{1}{a-1}(ax+b)^{1-1/a} & if $a\neq 0,1$,\\
        \log(x+b) & if $a=1$,\\
        -b\e^{-x/b} & if $a=0$, $b>0$
    \end{cases*}\label{eq:HARA}
\end{equation}
up to a positive affine transformation, which does not affect expected-utility rankings. Furthermore, the domain of $u$ is
\begin{equation*}
	C=\begin{cases*}
		(-b/a,\infty) & if $a>0$,\\
		(-\infty,\infty) & if $a=0$,\\
		(-\infty,-b/a) & if $a<0$.
	\end{cases*}
\end{equation*}
Every HARA utility in \eqref{eq:HARA} with $u'''>0$, equivalently with $a>-1$, satisfies Assumption~\ref{asmp:utility}. \citet{Kimball1990} defines absolute prudence as $-u'''/u''$; I call its reciprocal \emph{inverse absolute prudence} and denote it by
\begin{equation}
	\psi(x)\coloneqq-\frac{u''(x)}{u'''(x)},\qquad x\in C. \label{eq:psi}
\end{equation}
For HARA utility, \eqref{eq:HARA} implies
\begin{equation}
	\psi(x)=\frac{ax+b}{a+1}, \label{eq:HARA_inverseabsoluteprudence}
\end{equation}
which is affine.

In the borrowing constraint \eqref{eq:os_borrow}, $w_t$ is wealth at the beginning of time $t$. Here and throughout, \emph{wealth} means cash-on-hand after the current return and non-financial income are realized. The variable $b_t\in[0,\infty]$ is the borrowing limit. I set $b_T\equiv 0$, which rules out borrowing in the last period. In the budget constraint \eqref{eq:os_budget}, $R_t>0$ is the gross return on savings between time $t-1$ and time $t$, and $Y_t\ge 0$ is non-financial income that arrives at the beginning of time $t$. I use the following stochastic specification. Let $(\Omega,\cF,\mathrm{P})$ be a probability space, let $\set{\cF_t}_{t=0}^T$ be a filtration, and let $\E_t$ denote conditional expectation given $\cF_t$. The variables $\beta_t$, $R_t$, $Y_t$, and $b_t$ are $\cF_t$-measurable.

An adapted wealth-consumption process $\set{(w_t,c_t)}_{t=0}^T$ is \emph{feasible} if it satisfies the constraints \eqref{eq:os_feasible}--\eqref{eq:os_budget}. A process $\set{(w_t^*,c_t^*)}_{t=0}^T$ is \emph{optimal}, or a \emph{solution} to the optimal saving problem \eqref{eq:os}, if it maximizes the objective function \eqref{eq:os_obj} among all feasible processes. For each $t$ and history $\omega$, let $D_t(\omega)\subset\R$ be the wealth domain on which the time-$t$ continuation problem admits a solution, and define the joint domain
\begin{equation*}
    \cD_t\coloneqq\set{(w,\omega)\in\R\times\Omega:w\in D_t(\omega)}.
\end{equation*}
Throughout, I restrict attention to optimal saving problems for which the domains $D_t(\omega)$ are nonempty intervals and the conditional expectations used below are finite. With a slight abuse of notation, I call $c_t:\cD_t\to C$ a (time $t$) \emph{consumption function} if $c_t^*=c_t(w_t,\omega)$ represents the solution. Concavity of $c_t$ means that $c_t(\cdot,\omega)$ is concave on $D_t(\omega)$ for each history $\omega$. For a fixed history, I suppress the state argument and write $c_t(w)$.

\subsection{Main results}\label{subsec:results}

For compactness, I denote a class of optimal saving problems by $\OS(\bullet)$, where $\bullet$ indicates particular specifications such as $T=1$ (one-period model). A tilde indicates that a primitive can be stochastic, and restrictions following a semicolon further delimit the class. Thus $\OS(\tilde{\beta},\tilde{R},\tilde{Y},\tilde{b})$ denotes the general class of problems. If borrowing is freely allowed, so $b_t=\infty$ for all $t<T$, I write $\OS(\tilde{\beta},\tilde{R},\tilde{Y},\infty)$.

I first state the unrestricted benchmark in the present setting. Its sufficiency follows from the concavity-preservation argument of \citet{MaStachurskiToda2020JET}, and its necessity follows from the main result and HARA domain-translation argument of \citet{Toda2021JME}. Section~\ref{subsec:literature} discusses these connections. The benchmark provides a point of comparison for the main contribution below.

\begin{thm}[Concavity characterization: unrestricted case]\label{thm:unrestricted}
Under Assumption~\ref{asmp:utility}, the consumption functions are concave in every optimal saving problem in $\OS(\tilde{\beta},\tilde{R},\tilde{Y},\tilde{b})$ if and only if $u$ is HARA.
\end{thm}

The benchmark imposes no impatience condition, although the standing finiteness requirements on conditional expectations continue to apply. I next ask whether an economically relevant restriction on saving incentives permits uniform concavity beyond HARA.

\begin{defn}[Conditional impatience]
An optimal saving problem satisfies \emph{conditional impatience} if
\begin{equation}
    \E_t[\beta_{t+1}R_{t+1}]\le 1
    \quad\text{almost surely for all }t=0,\dots,T-1. \label{eq:conditionalImpatience}
\end{equation}
\end{defn}

This condition bounds only the conditional mean of the discounted gross return. It is weaker than the pointwise restriction $\beta_{t+1}R_{t+1}\le 1$ almost surely.

When the discount factor and gross return are constant, conditional impatience reduces to $\beta R\le 1$. Allowing income to be stochastic, \citet[Corollary~1]{ChamberlainWilson2000} show that $\beta R>1$ makes consumption grow without bound in an infinite-horizon optimal saving problem. Using this result, \citet[proof of Theorem~8]{StachurskiToda2019JET} show that market clearing requires $\beta R\le 1$ in a stationary equilibrium of a canonical Bewley--Huggett--Aiyagari economy with infinitely lived households, a common discount factor, a single risk-free asset, and a finite-mean wealth distribution.\footnote{For a related buffer-stock interpretation, \citet[Theorem~2]{Carroll2021} establishes a unique target ratio of cash-on-hand to permanent income in his CRRA model, conditional on a nondegenerate infinite-horizon solution and his other maintained assumptions. The ratio drifts toward the target in conditional expectation under a growth-impatience condition. With constant permanent income, this condition reduces to $\beta R<1$, while transitory income may remain stochastic.} The strict inequality $\beta R<1$ also appears in contraction results for the Coleman operator \citep[Assumption~2.1 and \S4.1]{LiStachurski2014} and in stationary-equilibrium results for a canonical Bewley economy with production \citep[Proposition~6 and Theorem~1]{Acikgoz2018}.

The following theorem is the main result of this paper.

\begin{thm}[Concavity characterization: conditional impatience]\label{thm:main}
Under Assumption~\ref{asmp:utility}, the consumption functions are concave in every optimal saving problem in
\begin{equation}
	\OS(\tilde{\beta},\tilde{R},\tilde{Y},\tilde{b};\ \E_t[\beta_{t+1}R_{t+1}]\le1\ \mathrm{a.s.},\ t=0,\dots,T-1) \label{eq:conditionallyImpatientClass}
\end{equation}
if and only if inverse absolute prudence, $-u''/u'''$, is concave.
\end{thm}

The unrestricted benchmark requires HARA utility, whose inverse absolute prudence is affine. Under conditional impatience, concavity of inverse absolute prudence is necessary and sufficient. The following example shows that this enlargement of the utility class is strict.

\begin{exmp}[Non-HARA perturbation of CRRA utility]\label{exmp:nonHARA}
Let $C=(0,\infty)$ and define
\begin{equation}
	u'(x)\coloneqq \int_x^\infty \gamma t^{-\gamma-1}\e^{-\delta t}\diff t, \label{eq:gamdel}
\end{equation}
where $\gamma>0$ and $\delta\ge 0$. When $\delta=0$, \eqref{eq:gamdel} gives $u'(x)=x^{-\gamma}$, which corresponds to CRRA utility. Clearly, $u'(0)=\infty$, $u'(\infty)=0$, so the Inada conditions \eqref{eq:Inada} hold. Moreover,
\begin{align*}
	u''(x)&=-\gamma x^{-\gamma-1}\e^{-\delta x}<0,\\
	u'''(x)&=\gamma x^{-\gamma-2}\e^{-\delta x}(\gamma+1+\delta x)>0.
\end{align*}
Inverse absolute prudence is then
\begin{equation*}
	\psi(x)\coloneqq -\frac{u''(x)}{u'''(x)}=\frac{x}{\gamma+1+\delta x},
\end{equation*}
which is affine when $\delta=0$ and strictly concave whenever $\delta>0$. Thus this family nests CRRA at $\delta=0$ and, for every $\delta>0$, provides a non-HARA utility satisfying the condition in Theorem~\ref{thm:main}.
\end{exmp}

\subsection{An analytical example of nonconcavity}\label{subsec:nonconcavity}

The proof of Theorem~\ref{thm:unrestricted} shows that uniform consumption concavity requires HARA utility even in deterministic one-period problems with unrestricted borrowing when $\beta$ and $R$ are unrestricted. Thus any non-HARA utility satisfying Assumption~\ref{asmp:utility} must generate a nonconcave consumption function in some such problem. This remains true even if inverse absolute prudence is concave. I now construct an analytical example.

Let $C=(0,\infty)$ and define
\begin{equation}
	u(x)\coloneqq(x+1)\log(x+1)-x\log x. \label{eq:nonconcavityUtility}
\end{equation}
Its first three derivatives are
\begin{align*}
	u'(x)&=\log(x+1)-\log x=\log(1+1/x)>0,\\
	u''(x)&=\frac{1}{x+1}-\frac{1}{x}=-\frac{1}{x(x+1)}<0,\\
	u'''(x)&=-\frac{1}{(x+1)^2}+\frac{1}{x^2}=\frac{2x+1}{x^2(x+1)^2}>0.
\end{align*}
Moreover, $u'(x)\to\infty$ as $x\downarrow0$ and $u'(x)\to0$ as $x\to\infty$, so $u$ satisfies Assumption~\ref{asmp:utility}. Its inverse absolute prudence \eqref{eq:psi} is
\begin{equation*}
	\psi(x)\coloneqq -\frac{u''(x)}{u'''(x)}=\frac{x(x+1)}{2x+1},
\end{equation*}
whose second derivative is $\psi''(x)=-2(2x+1)^{-3}<0$. Thus $\psi$ is strictly concave and satisfies the condition in Theorem~\ref{thm:main}. However, $u$ is not HARA: $\psi$ is not affine, whereas every HARA utility has affine inverse absolute prudence by \eqref{eq:HARA_inverseabsoluteprudence}.

Now fix deterministic constants $\beta,R>0$ and $Y\ge0$ satisfying $\beta R=2$. Consider the one-period problem with unrestricted borrowing. Given initial wealth $w>-Y/R$, the agent chooses $c_0,c_1>0$ to maximize $u(c_0)+\beta u(c_1)$ subject to
\begin{equation*}
	c_1=R(w-c_0)+Y.
\end{equation*}
Strict concavity and the Inada conditions imply a unique interior solution. The Euler equation and the expression for marginal utility above imply
\begin{equation*}
	\log(1+1/c_0)=2\log(1+1/c_1) \iff 	c_0=\frac{c_1^2}{2c_1+1}.
\end{equation*}
To solve for the consumption function, set $A\coloneqq w+Y/R$ and $c=c_0$. Then $c_1=R(A-c)$, so eliminating $c_1$ from the Euler equation yields
\begin{equation*}
	c=\frac{[R(A-c)]^2}{2R(A-c)+1}\iff R(R+2)c^2-(2R(R+1)A+1)c+R^2A^2=0.
\end{equation*}
At $c=0$, the quadratic polynomial equals $R^2A^2>0$, whereas at $c=A$, it equals $-A<0$. Since the leading coefficient is positive, its roots lie in $(0,A)$ and $(A,\infty)$. Feasibility requires $c=c_0\in(0,A)$; hence the smaller root gives the consumption function
\begin{equation}
	c(w)=\frac{2R(R+1)A+1-\sqrt{\Delta(A)}}{2R(R+2)}, \label{eq:nonconcavityPolicy}
\end{equation}
where $\Delta(A)\coloneqq 4R^2A^2+4R(R+1)A+1$. Differentiating with respect to $A$ gives
\begin{align*}
	\frac{\diff}{\diff A}\Delta(A)^{1/2}&=\Delta(A)^{-1/2}2R(2RA+R+1),\\
	\frac{\diff^2}{\diff A^2}\Delta(A)^{1/2}&=\Delta(A)^{-1/2}(2R)^2-\Delta(A)^{-3/2}[2R(2RA+R+1)]^2\\
	&=-\Delta(A)^{-3/2}4R^3(R+2).
\end{align*}
Because $A=w+Y/R$, we have $\diff A/\diff w=1$, so the second derivative of the consumption function is
\begin{equation*}
	c''(w)=-\frac{1}{2R(R+2)}\frac{\diff^2}{\diff A^2}\Delta(A)^{1/2}=2R^2\Delta(A)^{-3/2}>0.
\end{equation*}
Thus the consumption function is strictly convex, and hence not concave, for every deterministic $(\beta,R,Y)$ satisfying $\beta R=2$.

For illustration, I set $\beta=1$, $R=2$, and $Y=1$. Then the feasible wealth domain is $w>-Y/R=-1/2$. Figure~\ref{fig:nonconcavity} plots the consumption function \eqref{eq:nonconcavityPolicy} and its second derivative for these parameter values.

\begin{figure}[!htb]
	\centering
	\def\ncReturn{2}
	\def\ncIncome{1}
	\pgfmathsetmacro{\ncLower}{-\ncIncome/\ncReturn}
	\pgfmathsetmacro{\ncStart}{\ncLower+0.0001}
	\definecolor{matlabblue}{rgb}{0,0.4470,0.7410}
	\begin{tikzpicture}[
		font=\footnotesize,
		declare function={
			ncA(\x)=\x+\ncIncome/\ncReturn;
			ncDelta(\a)=4*\ncReturn^2*\a^2+4*\ncReturn*(\ncReturn+1)*\a+1;
			ncPolicy(\x)=2*\ncReturn^2*ncA(\x)^2/(2*\ncReturn*(\ncReturn+1)*ncA(\x)+1+sqrt(ncDelta(ncA(\x))));
			ncSecond(\x)=2*\ncReturn^2/ncDelta(ncA(\x))^(3/2);
		}
	]
		\begin{groupplot}[
			group style={group size=2 by 1, group name=nonconcavity, horizontal sep=0.13\textwidth},
			width=0.38\textwidth, height=0.2755\textwidth, scale only axis,
			xmin=\ncLower, xmax=1, domain=\ncStart:1, samples=200,
			ymin=0, xtick={-0.5,0,0.5,1},
			xlabel={Wealth $w$},
			axis lines=box, tick align=inside, tick pos=both,
			scaled ticks=false,
			tick label style={/pgf/number format/fixed, /pgf/number format/precision=1}
		]
			\nextgroupplot[ymax=0.8, ytick={0,0.2,0.4,0.6,0.8},
				ylabel={Consumption function $c(w)$}]
			\addplot[matlabblue, line width=0.8pt, no marks] {ncPolicy(x)};
			\nextgroupplot[ymax=8, ytick={0,2,4,6,8},
				ylabel={Second derivative $c''(w)$}]
			\addplot[matlabblue, line width=0.8pt, no marks] {ncSecond(x)};
		\end{groupplot}
	\end{tikzpicture}
	\caption{Consumption and its second derivative in the deterministic one-period problem with $u(x)=(x+1)\log(x+1)-x\log x$, $\beta=1$, $R=2$, and $Y=1$. The consumption function is strictly convex on its feasible wealth domain $w>-1/2$.}
	\label{fig:nonconcavity}
\end{figure}
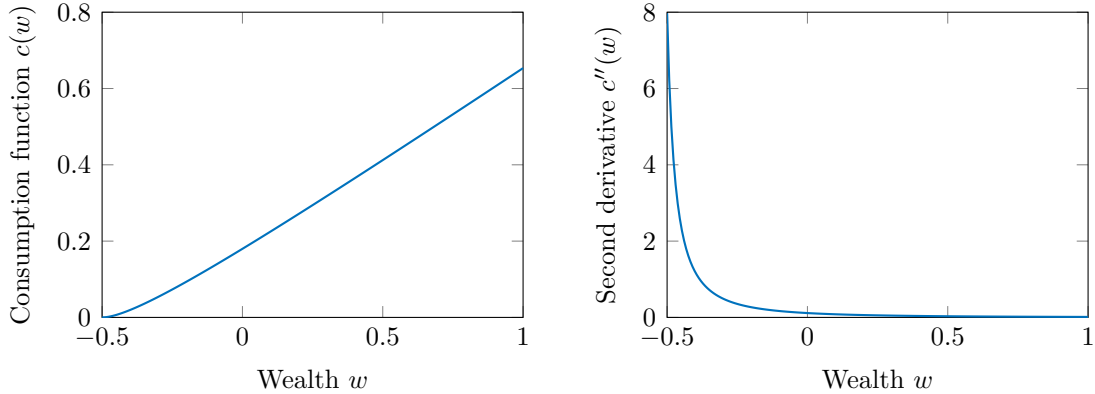

Theorem~\ref{thm:main} guarantees concave consumption functions for this utility in every problem satisfying conditional impatience. Here, $\beta R=2$ violates that condition and produces a strictly convex consumption function, even without uncertainty or borrowing constraints. Thus concave inverse absolute prudence no longer suffices once discounted returns are unrestricted, illustrating the distinction between Theorems~\ref{thm:unrestricted} and \ref{thm:main}. We revisit this example in Section~\ref{sec:perfectForesight}.

\subsection{Related literature}\label{subsec:literature}

Table~\ref{tab:literature} compares the results most directly related to the global concavity characterizations in this paper. I omit standard regularity conditions and work primarily concerned with local concavification, the comparative effects of risk or liquidity constraints, or precautionary saving; I discuss selected contributions in those strands below.

\begin{table}[!htb]
\centering
\caption{Selected results on global consumption concavity.}
\label{tab:literature}
\footnotesize
\renewcommand{\arraystretch}{1.15}
\begin{tabularx}{\textwidth}{@{}>{\raggedright\arraybackslash}p{0.16\textwidth}>{\raggedright\arraybackslash}p{0.22\textwidth}>{\raggedright\arraybackslash}p{0.15\textwidth}>{\raggedright\arraybackslash}p{0.21\textwidth}>{\raggedright\arraybackslash}X@{}}
\toprule
Result & Environment & Borrowing & Key condition & Conclusion \\
\midrule
\citet{CarrollKimball1996}
& Finite horizon; stochastic $\beta$, $R$, and $Y$
& Unrestricted
& HARA utility
& Concavity in wealth \\

\citet{SuenConcave}
& Finite horizon; constant $\beta$ and $R$; permanent and transitory income risk
& Deterministic, age-dependent limits
& $\beta R\le 1$ and concavity of $-u''/u'''$
& Joint concavity in assets and one income component \\

\citet{PelusoTrannoy2012}
& One period; deterministic $\beta$ and $R$
& Unrestricted
& $\beta R\ge1$ and convexity of $-u'/u''$
& Concavity in wealth \\

\citet*{MaStachurskiToda2020JET}
& Infinite-horizon Markov model; stochastic $\beta$, $R$, and $Y$
& No borrowing
& CRRA utility
& Concavity in wealth \\

\citet{Toda2021JME}
& One period; finite-support stochastic $\beta$, $R$, and $Y$
& Unrestricted
& Concavity in every such problem
& Necessity of CRRA utility on $(0,\infty)$ \\

\citet*{CarrollHolmKimball2021}
& Finite horizon; constant $\beta$ and $R$; future risks
& Period-by-period liquidity constraints
& $u'''>0$ and decreasing absolute prudence
& Concavity with future risks and constraints \\

\midrule

Theorem~\ref{thm:unrestricted}
& Stochastic $\beta$, $R$, $Y$, and $b$
& Stochastic limits
& No return restriction
& iff HARA \\

Theorem~\ref{thm:main}
& Stochastic $\beta$, $R$, $Y$, and $b$
& Stochastic limits
& $\E_t[\beta_{t+1}R_{t+1}]\le1$
& iff concave $\psi$ \\

Proposition~\ref{prop:deterministic}
& Deterministic $\beta$, $R$, and $Y$
& Unrestricted
& \begin{tabular}[t]{@{}l@{}}
    $\beta_tR_t\le1$ \\
    $\beta_tR_t\ge1$
  \end{tabular}
& \begin{tabular}[t]{@{}l@{}}
    iff concave $\tau$ \\
    iff convex $\tau$
  \end{tabular} \\

Proposition~\ref{prop:reverseImpatience}
& Stochastic $\beta$, $R$, and $Y$
& Unrestricted
& \begin{tabular}[t]{@{}l@{}}
    $\beta_tR_t\le1$ \\
    $\beta_tR_t\ge1$
  \end{tabular}
& \begin{tabular}[t]{@{}l@{}}
    iff concave $\psi$ \\
    iff HARA
  \end{tabular} \\
\bottomrule
\end{tabularx}
\captionsetup{font=footnotesize,justification=raggedright,singlelinecheck=false,position=bottom}
\caption*{\textit{Notes:} In the lower block, all problems have a finite horizon, and ``iff'' refers to consumption concavity in every problem in the indicated class. Here $\tau=-u'/u''$ is absolute risk tolerance and $\psi=-u''/u'''$ is inverse absolute prudence. Return restrictions hold at every transition and almost surely under uncertainty. The converses for Proposition~\ref{prop:reverseImpatience} follow from Theorems~\ref{thm:unrestricted} and \ref{thm:main}.}
\end{table}

The unrestricted benchmark in Theorem~\ref{thm:unrestricted} combines earlier sufficiency and necessity arguments. \citet[Theorem~1]{CarrollKimball1996} show that HARA utility guarantees consumption concavity in a finite-horizon model with stochastic discounting, returns, and income under unrestricted borrowing. \citet*[Proposition~2.5]{MaStachurskiToda2020JET} establish concavity under an abstract preservation condition in an infinite-horizon Markov model with stochastic discounting, returns, and income but no borrowing. Their Remark~2.1 verifies the condition for CRRA utility. As \citet[footnote~1 and \S3]{Toda2021JME} notes, the same preservation argument applies to finite horizons and liquidity constraints and, after translating the consumption domain, to HARA utility. Conversely, \citet[Theorem~4]{Toda2021JME} proves that CRRA is necessary on $(0,\infty)$ if the one-period consumption function is concave for every strictly positive, finite-support random vector $(\beta,R,Y)$; translating the domain yields HARA more generally.

The closest predecessor to Theorem~\ref{thm:main} is the unpublished manuscript by \citet{SuenConcave}. In a finite-horizon life-cycle model with a constant discount factor and risk-free return, deterministic age-dependent borrowing limits, and independent permanent and transitory income shocks with compact support, his Theorem~8 proves that $\beta R\le1$ and concavity of $-u''/u'''$ imply joint concavity in assets and one income component, holding the other fixed. Because the return is deterministic, assets and wealth are related affinely conditional on the income state, so this result implies concavity in wealth. Theorem~\ref{thm:main} establishes necessity for uniform consumption concavity in wealth and extends sufficiency to jointly stochastic discount factors, returns, income, and borrowing limits, without requiring Suen's permanent--transitory decomposition, independence, or compact-support assumptions. Suen's result remains stronger along the income dimensions.

Related analytical work derives consumption rules or constraint-induced curvature under more specific deterministic environments. \citet{Park2006} obtains a closed-form inverse consumption function under CRRA utility in continuous time. \citet{Holm2018} obtains an analytical characterization under HARA utility and proves strict concavity when a liquidity constraint is relevant, without restricting the sign of $u'''$. \citet{Roulleau-Pasdeloup2026} derives a global explicit consumption function under CRRA utility at a zero net interest rate and an explicit approximation for rates near zero. In a deterministic infinite-horizon model, \citet*{NishiyamaKato2012} show that a liquidity constraint can produce a concave consumption function with quadratic utility, for which $u'''=0$.

Other papers study how risk and liquidity alter consumption curvature. In a two-period model, \citet{Kimball1990} introduces absolute prudence as an index of precautionary saving and shows how its monotonicity determines the effect of income risk on the marginal propensity to consume. In a finite-horizon discrete-time model, \citet*{CarrollHolmKimball2021} show that, when $u'''>0$ and absolute prudence is decreasing, future risks and liquidity constraints make the consumption function more concave. These papers identify mechanisms through which particular risks and constraints create concavity, whereas Theorem~\ref{thm:main} characterizes the utility functions that guarantee concavity throughout the conditionally impatient stochastic class \eqref{eq:conditionallyImpatientClass}.

Taken together, Theorems~\ref{thm:unrestricted} and \ref{thm:main} separate the roles of saving incentives and preference curvature: conditional impatience enlarges the admissible preference class from HARA to utilities with concave inverse absolute prudence.

\section{Perfect foresight and the role of uncertainty}\label{sec:perfectForesight}

Theorems~\ref{thm:unrestricted} and \ref{thm:main} allow discount factors, returns, income, and borrowing limits to be jointly stochastic. To isolate the role of uncertainty, this section examines perfect foresight with unrestricted borrowing. Under the restrictions on discounted returns considered below, I characterize consumption curvature in terms of the curvature of absolute risk tolerance. I then revisit the example in Section~\ref{subsec:nonconcavity}, correct a claim in the literature, and show that the perfect-foresight characterizations do not extend to uncertainty even when the same restrictions on discounted returns hold almost surely.

\subsection{Consumption curvature under perfect foresight}\label{subsec:deterministic}

Consider optimal saving problems in $\OS(\beta,R,Y,\infty)$, where the absence of tildes indicates perfect foresight and $\infty$ indicates the agent can borrow without restriction. Thus, the agent knows the entire paths of $\beta_t$, $R_t$, and $Y_t$, although these variables may vary over time. Define \emph{absolute risk tolerance} by
\begin{equation}
	\tau(x)\coloneqq-\frac{u'(x)}{u''(x)},\qquad x\in C. \label{eq:riskTolerance}
\end{equation}
The following proposition characterizes consumption curvature in this setting.

\begin{prop}[Perfect-foresight characterization]\label{prop:deterministic}
Under Assumption~\ref{asmp:utility}, the following statements hold.
\begin{enumerate}
	\item\label{item:deterministicImpatient} The consumption functions are concave in every optimal saving problem in
	\begin{equation*}
		\OS(\beta,R,Y,\infty;\ \beta_tR_t\le1,\ t=1,\dots,T)
	\end{equation*}
	if and only if $\tau$ is concave on $C$.
	\item\label{item:deterministicPatient} The consumption functions are concave in every optimal saving problem in
	\begin{equation*}
		\OS(\beta,R,Y,\infty;\ \beta_tR_t\ge1,\ t=1,\dots,T)
	\end{equation*}
	if and only if $\tau$ is convex on $C$.
\end{enumerate}
Replacing consumption concavity by consumption convexity reverses the curvature condition on $\tau$ in both statements.
\end{prop}

In a two-good model with additively separable utility and the same utility index for both goods, \citet[Proposition~3(iii)]{PelusoTrannoy2012} show that demand for the good with lower consumption is concave in wealth in every such problem if and only if risk tolerance is convex. Their Section~3 applies this result to intertemporal consumption when $\beta R>1$, corresponding to Proposition~\ref{prop:deterministic}\ref{item:deterministicPatient} with $T=1$ and strict inequality. Their Proposition~6 extends the sufficiency argument to a finite horizon. Proposition~\ref{prop:deterministic} provides finite-horizon characterizations of both consumption concavity and convexity under either $\beta_tR_t\le1$ at every transition or $\beta_tR_t\ge1$ at every transition.

The proposition and its proof yield the following characterization of affine consumption functions without restrictions on discounted returns.

\begin{cor}[Affine consumption]\label{cor:deterministicAffine}
Under Assumption~\ref{asmp:utility}, the consumption functions are affine in every optimal saving problem in $\OS(\beta,R,Y,\infty)$ if and only if $u$ is HARA.
\end{cor}

In Proposition~\ref{prop:deterministic}, uniform consumption concavity requires concave risk tolerance when $\beta_tR_t\le1$ at every transition, but convex risk tolerance when $\beta_tR_t\ge1$ at every transition. The Euler equation
\begin{equation*}
	u'(c_t)=\beta_{t+1}R_{t+1}u'(c_{t+1})
\end{equation*}
helps explain why the two cases require opposite curvature conditions. Because marginal utility decreases with consumption, $\beta_{t+1}R_{t+1}\le1$ implies $c_{t+1}\le c_t$, whereas $\beta_{t+1}R_{t+1}\ge1$ implies $c_{t+1}\ge c_t$. This ordering determines whether concave or convex risk tolerance yields concave consumption functions. With unrestricted borrowing, deterministic income shifts the present value of resources without changing the conclusion about curvature.

In each case of Proposition~\ref{prop:deterministic}, strict curvature of $\tau$ gives the corresponding strict curvature of $c_t$ for every $t<T$ if the return inequality is also strict at every transition: $\beta_tR_t<1$ in part~\ref{item:deterministicImpatient} and $\beta_tR_t>1$ in part~\ref{item:deterministicPatient}. Corollary~\ref{cor:deterministicAffine} identifies HARA as exactly the class that gives affine consumption functions in every perfect-foresight problem with unrestricted borrowing. The equality $\beta_tR_t=1$ at every transition also gives affine consumption functions for any utility satisfying Assumption~\ref{asmp:utility}. Thus weak curvature assumptions alone do not guarantee strict concavity or convexity. I prove the proposition, the corollary, and these additional claims in Section~\ref{subsec:proofComparisons}.

The example in Section~\ref{subsec:nonconcavity} illustrates the convexity conclusion of Proposition~\ref{prop:deterministic}\ref{item:deterministicPatient}. For the utility function \eqref{eq:nonconcavityUtility}, we have
\begin{align*}
	\tau(x)&=x(x+1)\log(1+1/x),\\
	\tau'(x)&=(2x+1)\log(1+1/x)-1,\\
	\tau''(x)&=2\log(1+1/x)-\frac{2x+1}{x(x+1)},\\
	\tau'''(x)&=\frac{1}{x^2(x+1)^2}>0.
\end{align*}
Since $\tau'''>0$ and $\lim_{x\to\infty}\tau''(x)=0$, we have $\tau''(x)<0$ for every $x>0$. Thus risk tolerance is strictly concave. Because $\beta R=2>1$ in the example, Proposition~\ref{prop:deterministic}\ref{item:deterministicPatient}, together with the strict-curvature observation above, implies that $c_0$ is strictly convex. This conclusion does not require the closed-form solution~\eqref{eq:nonconcavityPolicy}.

\citet*[Proposition~1]{GongZhongZou2012} claim that concave risk tolerance yields concave consumption functions in deterministic problems with a finite horizon, unrestricted borrowing, and $\beta_tR_t>1$ at every transition. The example contradicts this claim.\footnote{They reverse an inequality in the proof of their Lemma~3. Their Lemma~1 gives $c_{T-1}<c_T$ when $\beta_TR_T>1$, so concavity of $\tau$ implies $\tau'(c_T)\le\tau'(c_{T-1})$. In the first display on page~101, they take a weighted average of $\tau'(c_{T-1})+1$ and $\tau'(c_T)+1$ with positive weights that sum to one. This average cannot exceed $\tau'(c_{T-1})+1$, but they claim that it is strictly larger. Their Remark~1 also claims strict consumption concavity under a weak curvature assumption, but HARA utilities give affine consumption functions.} Proposition~\ref{prop:deterministic} gives two corrections: require convex risk tolerance when $\beta_tR_t>1$, or retain concave risk tolerance and impose $\beta_tR_t<1$. Keeping both original assumptions instead yields convex consumption functions.

\subsection{Comparison with uncertainty}\label{subsec:uncertaintyComparison}

Proposition~\ref{prop:deterministic} depends on perfect foresight. Under uncertainty, the Euler equation averages marginal utility across future states. Even a bound on $\beta_{t+1}R_{t+1}$ that holds in every state does not determine whether future consumption lies above or below current consumption in each state. The deterministic argument based on the ordering of current and future consumption therefore does not apply.

The following proposition gives the corresponding necessary conditions under uncertainty. I state these results for problems with unrestricted borrowing. Theorems~\ref{thm:unrestricted} and \ref{thm:main} establish sufficiency even with stochastic borrowing limits.

\begin{prop}[Pathwise restrictions]\label{prop:reverseImpatience}
Under Assumption~\ref{asmp:utility}, the following statements hold.
\begin{enumerate}
	\item\label{item:pathwisePatient} If the consumption functions are concave in every optimal saving problem in
	\begin{equation*}
		\OS(\tilde{\beta},\tilde{R},\tilde{Y},\infty;\ \beta_{t+1}R_{t+1}\ge1\ \mathrm{a.s.},\ t=0,\dots,T-1),
	\end{equation*}
	then $u$ is HARA.
	\item\label{item:pathwiseImpatient} If the consumption functions are concave in every optimal saving problem in
	\begin{equation*}
		\OS(\tilde{\beta},\tilde{R},\tilde{Y},\infty;\ \beta_{t+1}R_{t+1}\le1\ \mathrm{a.s.},\ t=0,\dots,T-1),
	\end{equation*}
	then inverse absolute prudence $\psi=-u''/u'''$ is concave on $C$.
\end{enumerate}
\end{prop}

Theorems~\ref{thm:unrestricted} and \ref{thm:main} give the converses in parts~\ref{item:pathwisePatient} and \ref{item:pathwiseImpatient}, respectively. Part~\ref{item:pathwisePatient} also implies that uniform consumption concavity over all problems satisfying
\begin{equation*}
	\E_t[\beta_{t+1}R_{t+1}]\ge1
	\quad\text{almost surely for all }t=0,\dots,T-1
\end{equation*}
requires HARA, because that class contains the problems in part~\ref{item:pathwisePatient}.

Table~\ref{tab:perfectForesightComparison} compares the necessary and sufficient conditions for uniform consumption concavity under the same restrictions on discounted returns. Each entry requires concavity in every problem in the indicated class, not merely for one path or distribution of shocks. When $\beta_tR_t\le1$, allowing uncertainty changes the condition from concave risk tolerance to concave inverse absolute prudence. When $\beta_tR_t\ge1$, uncertainty instead requires HARA rather than convex risk tolerance. These distinctions remain even when the return restrictions hold almost surely.

\begin{table}[!htb]
	\centering
	\caption{Necessary and sufficient conditions for uniform consumption concavity with unrestricted borrowing. Return restrictions hold at every transition and almost surely under uncertainty.}
	\label{tab:perfectForesightComparison}
	\begin{tabular}{@{}lll@{}}
		\toprule
		Return restriction & Perfect foresight & Uncertainty \\
		\midrule
		$\beta_tR_t\le1$ & Concave $\tau=-u'/u''$ & Concave $\psi=-u''/u'''$ \\
		$\beta_tR_t\ge1$ & Convex $\tau=-u'/u''$ & HARA \\
		\bottomrule
	\end{tabular}
\end{table}

\section{Conclusion}\label{sec:conclusion}

This paper characterizes the preferences that make consumption concave in wealth across finite-horizon optimal saving problems. Under conditional impatience, concavity of inverse absolute prudence is necessary and sufficient, even with jointly stochastic discount factors, returns, income, and borrowing limits. Without a restriction on discounted returns, uniform concavity requires HARA utility. Conditional impatience therefore permits declining marginal propensities to consume for a strictly larger class of preferences.

The perfect-foresight results isolate the role of uncertainty. With unrestricted borrowing, absolute risk tolerance characterizes consumption curvature under restrictions on discounted returns. Under uncertainty, marginal-utility aggregation instead makes inverse absolute prudence relevant under conditional impatience and requires HARA for uniform concavity when the discounted gross return is at least one almost surely. These characterizations provide primitive conditions for declining marginal propensities to consume that can be checked without solving the household problem.

\appendix

\section{Proofs}\label{sec:proof}

To motivate the proof technique of Theorem~\ref{thm:main}, consider a one-period problem with unrestricted borrowing and let $s$ denote saving. Current and future consumption are $w-s$ and $Rs+Y$, respectively. The first-order condition gives
\begin{equation*}
    u'(w-s)=\E[\beta R u'(Rs+Y)]\iff c=w-s=(u')^{-1}(\E[\beta Ru'(Rs+Y)]),
\end{equation*}
where $u''<0$ makes $u'$ invertible.

Accordingly, given any random vector $(\beta,R,Y)\gg 0$ and letting $\phi=u'$ be the marginal utility function, define the function
\begin{equation}
    g(s)\coloneqq \phi^{-1}(\E[\beta R \phi(Rs+Y)]). \label{eq:g}
\end{equation}
The function $g$ in \eqref{eq:g} returns the consumption level $c\in C$ that implies the saving $s$ in the one-period problem. Because the domain of $u$ (and hence $\phi=u'$) is $C\subset \R$, defining $g(s)$ requires $Rs+Y\in C$ almost surely and $\E[\beta R\phi(Rs+Y)]<\infty$. The Inada conditions \eqref{eq:Inada} then imply that $g(s)$ is well-defined and belongs to $C$. For finite-support shocks, $g$ is as smooth as $\phi$ on its domain. The argument below does not require differentiating this expectation.

The following lemma provides the first step for proving necessity. It allows us to pass from the concavity of a consumption function to the concavity of the corresponding function $g$ in \eqref{eq:g}.

\begin{lem}[Concavity transfer]\label{lem:g}
Under Assumption~\ref{asmp:utility}, consider a one-period optimal saving problem with unrestricted borrowing in which $(\beta,R,Y)$ has finite support. If the consumption function is concave, then the corresponding function $g$ in \eqref{eq:g} is concave on its domain.
\end{lem}

\begin{proof}
This argument adapts \citet[Lemma~2]{Toda2021JME}. Because the consumption domain here is an arbitrary open interval, I first verify that the saving levels at which $g$ is defined are exactly those attained by the optimal saving function. Let $S$ denote the domain of $g$. Because $(\beta,R,Y)$ has finite support and $C$ is an open interval, $S$ is an open interval. Moreover, $g$ is strictly increasing: if $s_1<s_2$, then $Rs_1+Y<Rs_2+Y$ in every state, so the expectation in \eqref{eq:g} strictly decreases, and applying the decreasing function $\phi^{-1}$ reverses the inequality.

For any $s\in S$, define $w\coloneqq s+g(s)$. At this wealth, $s$ is feasible and satisfies the first-order condition
\begin{equation*}
	\phi(w-s)=\E[\beta R\phi(Rs+Y)].
\end{equation*}
The one-period objective is strictly concave in saving, so $s$ is the unique optimum. Thus $w$ belongs to the wealth domain, and, writing $s(w)\coloneqq w-c(w)$ for optimal saving, we have
\begin{equation}
	s(w)=s,\qquad c(w)=g(s). \label{eq:consumptionSavingLink}
\end{equation}
Conversely, the Euler equation at any optimum implies $g(s(w))=c(w)$. Hence $S$ is exactly the range of the optimal saving function.

Take $s_1,s_2\in S$ and $\alpha\in[0,1]$. Define
\begin{equation*}
	s_\alpha\coloneqq(1-\alpha)s_1+\alpha s_2,\qquad
	\bar c\coloneqq(1-\alpha)g(s_1)+\alpha g(s_2),
\end{equation*}
and let $w_i\coloneqq s_i+g(s_i)$ for $i=1,2$ and $\bar w\coloneqq(1-\alpha)w_1+\alpha w_2=s_\alpha+\bar c$. The maintained assumption that the wealth domain is an interval implies that $\bar w$ belongs to it. By the concavity of $c$ and \eqref{eq:consumptionSavingLink}, we have
\begin{equation*}
	c(\bar w)\ge(1-\alpha)c(w_1)+\alpha c(w_2)=\bar c.
\end{equation*}
Therefore $s(\bar w)=\bar w-c(\bar w)\le s_\alpha$. Since $g$ is increasing and $g(s(\bar w))=c(\bar w)$, it follows that $g(s_\alpha)\ge g(s(\bar w))=c(\bar w)\ge\bar c$. Thus $g$ is concave on $S$.
\end{proof}

\subsection{Aggregation lemmas}\label{subsec:proofAggregation}

The finite-dimensional aggregation argument underlying the following lemma appears in Sections 3.1--3.2 of \citet{SuenConcave}. He studies an aggregator with weights of the form $\beta R P_n$ and proves its concavity by a Hessian and Cauchy--Schwarz argument. The following lemma reformulates that argument for arbitrary positive weights with total mass at most one and adds an equivalence with the one-dimensional restrictions used in Lemma 3 of \citet{Toda2021JME}. In the present stochastic setting, the restriction on the sum of the weights corresponds to conditional impatience \eqref{eq:conditionalImpatience}.

\begin{lem}[Finite-dimensional aggregator]\label{lem:Phi}
Under Assumption~\ref{asmp:utility}, let $\phi\coloneqq u'$. Since $\phi'=u''<0$ and the Inada conditions imply $\phi(C)=(0,\infty)$, the function $\phi$ is a strictly decreasing bijection from $C$ onto $(0,\infty)$. Hence define $\Phi:(0,\infty)\to(0,\infty)$ by
\begin{equation*}
	\Phi(m)\coloneqq \frac{[\phi'(\phi^{-1}(m))]^2}{\phi''(\phi^{-1}(m))}.
\end{equation*}
For $N\in\N$, $p\in\R_{++}^N$ such that $\sum_{n=1}^Np_n\le 1$, and $x\in C^N$, define the finite-dimensional aggregator
\begin{equation}
	\cA_p(x)\coloneqq\phi^{-1}\left(\sum_{n=1}^Np_n\phi(x_n)\right). \label{eq:multivariateAggregator}
\end{equation}
For $v\in\R^N$, define its one-dimensional restriction
\begin{equation}
	G(s;p,x,v)\coloneqq\cA_p(x+vs)
	=\phi^{-1}\left(\sum_{n=1}^Np_n\phi(x_n+v_ns)\right) \label{eq:oneDimensionalAggregator}
\end{equation}
on the open interval
\begin{equation}
	D(x,v)\coloneqq \set{s\in\R:x_n+v_ns\in C\text{ for all }n=1,\dots,N}. \label{eq:oneDimensionalDomain}
\end{equation}
Then the following statements are equivalent.
\begin{enumerate}
	\item\label{item:Phi1} The function $G(\cdot;p,x,v)$ is concave on $D(x,v)$ for every $N,p,x$ as above and every $v\in\R_{++}^N$.
	\item\label{item:Phi2} The function $\cA_p$ is concave on $C^N$ for every $N$ and $p$ as above.
	\item\label{item:Phi3} The function $\Phi$ is concave on $(0,\infty)$.
\end{enumerate}
\end{lem}

\begin{proof}
Clearly, \ref{item:Phi2}$\Rightarrow$\ref{item:Phi1}. The implication \ref{item:Phi1}$\Rightarrow$\ref{item:Phi3}, which recovers concavity of $\Phi$ from its one-dimensional restrictions, adapts Lemma 3 of \citet{Toda2021JME}, which uses the technique in \citet[\S3.16]{HardyLittlewoodPolyaInequalities}. The implication \ref{item:Phi3}$\Rightarrow$\ref{item:Phi2} adapts the Hessian and Cauchy--Schwarz argument in Section 3.2 of \citet{SuenConcave}.

\medskip
\noindent \ref{item:Phi1}$\Rightarrow$\ref{item:Phi3} Suppose \ref{item:Phi1} holds, so $G$ is concave. Because the Inada conditions imply $\phi(C)=(0,\infty)$, the functions $\cA_p$ and $G$ in \eqref{eq:multivariateAggregator} and \eqref{eq:oneDimensionalAggregator} are well-defined. Fix $N,p,x,v$ as in \ref{item:Phi1} and $s\in D(x,v)$. To simplify notation, write $G=G(s;p,x,v)$, $z_n=x_n+v_ns$, and $\sum=\sum_{n=1}^N$. Differentiating
\begin{equation*}
	\phi(G)=\sum p_n\phi(z_n)
\end{equation*}
with respect to $s$ twice, we obtain
\begin{align*}
	\phi'(G)G'&=\sum p_nv_n\phi'(z_n),\\
	\phi''(G)(G')^2+\phi'(G)G''&=\sum p_nv_n^2\phi''(z_n).
\end{align*}
Eliminating $G'$ yields
\begin{equation}
	\phi'(G)^3G''=\phi'(G)^2\sum p_nv_n^2\phi''(z_n)-\phi''(G)\left(\sum p_nv_n\phi'(z_n)\right)^2. \label{eq:GsecondDerivative}
\end{equation}
Since $\phi'=u''<0$ and $\phi''=u'''>0$, it follows that $G''(s)\le 0$ if and only if
\begin{equation}
	\Phi\left(\sum p_n\phi(z_n)\right)\ge \frac{\left(\sum p_nv_n\phi'(z_n)\right)^2}{\sum p_nv_n^2\phi''(z_n)}. \label{eq:Gconcave}
\end{equation}
We now apply the Cauchy--Schwarz inequality. To this end, define
\begin{equation*}
	a_n\coloneqq\sqrt{p_n\phi''(z_n)}v_n
	\quad\text{and}\quad
	b_n\coloneqq\frac{\sqrt{p_n}\phi'(z_n)}{\sqrt{\phi''(z_n)}}.
\end{equation*}
Then $a_nb_n=p_nv_n\phi'(z_n)$, and hence
\begin{align*}
	\left(\sum p_nv_n\phi'(z_n)\right)^2
	&=\left(\sum a_nb_n\right)^2\le\left(\sum a_n^2\right)\left(\sum b_n^2\right)\\
	&=\left(\sum p_nv_n^2\phi''(z_n)\right)
	\left(\sum p_n\frac{\phi'(z_n)^2}{\phi''(z_n)}\right).
\end{align*}
Since the first factor on the last line is strictly positive, dividing by it gives
\begin{equation}
	\frac{\left(\sum p_nv_n\phi'(z_n)\right)^2}{\sum p_nv_n^2\phi''(z_n)}
	\le \sum p_n\frac{\phi'(z_n)^2}{\phi''(z_n)}
	=\sum p_n\Phi(\phi(z_n)). \label{eq:CS}
\end{equation}
Equality in \eqref{eq:CS} holds if the vectors $(a_n)$ and $(b_n)$ are proportional. Because $a_n>0$ and $b_n<0$, this is equivalent to
\begin{equation}
	v_n=-k\frac{\phi'(z_n)}{\phi''(z_n)}\quad\text{for all }n \label{eq:CSequality}
\end{equation}
for some $k>0$.

Let $p\in\R_{++}^N$ satisfy $\sum p_n=1$ and let $m\in\R_{++}^N$ be arbitrary. Set $x_n=\phi^{-1}(m_n)$ and
\begin{equation*}
	v_n=-\frac{\phi'(x_n)}{\phi''(x_n)}>0.
\end{equation*}
Evaluating \eqref{eq:Gconcave} at $s=0$ and using the equality condition in \eqref{eq:CS}, we obtain
\begin{equation*}
	\Phi\left(\sum p_nm_n\right)\ge \sum p_n\Phi(m_n).
\end{equation*}
Since $N,p,m$ are arbitrary, $\Phi$ is concave on $(0,\infty)$, so \ref{item:Phi3} holds.

\medskip
\noindent \ref{item:Phi3}$\Rightarrow$\ref{item:Phi2} Suppose \ref{item:Phi3} holds, so $\Phi$ is concave. Fix $N$ and $p$ as in the statement of the lemma, $x\in C^N$, and an arbitrary $v\in\R^N$. If $v=0$, then $G(\cdot;p,x,v)$ is constant and hence concave. Therefore, suppose $v\neq 0$, so the denominator in \eqref{eq:Gconcave} is strictly positive. Let $s\in D(x,v)$, with $D$ as in \eqref{eq:oneDimensionalDomain}, and set $z_n\coloneqq x_n+v_ns$ and $\lambda\coloneqq\sum p_n\le 1$. Define $q_n\coloneqq p_n/\lambda$, $m_n\coloneqq\phi(z_n)$, and $\bar{m}\coloneqq\sum q_nm_n$. If $\lambda<1$, then for every $\epsilon>0$, the concavity and positivity of $\Phi$ imply
\begin{equation*}
	\Phi(\lambda\bar{m}+(1-\lambda)\epsilon)\ge \lambda\Phi(\bar{m})+(1-\lambda)\Phi(\epsilon)\ge \lambda\Phi(\bar{m}).
\end{equation*}
Letting $\epsilon\downarrow 0$ and using the continuity of $\Phi$, we obtain $\Phi(\lambda\bar{m})\ge \lambda\Phi(\bar{m})$. This inequality is immediate if $\lambda=1$. Therefore, using the concavity of $\Phi$ once more, we obtain
\begin{align*}
	\Phi\left(\sum p_nm_n\right)&=\Phi(\lambda\bar{m})\ge \lambda\Phi(\bar{m})\\
	&\ge \lambda\sum q_n\Phi(m_n)=\sum p_n\Phi(m_n).
\end{align*}
The Cauchy--Schwarz inequality \eqref{eq:CS}, which remains valid for arbitrary $v\neq 0$, therefore yields
\begin{equation*}
	\Phi\left(\sum p_nm_n\right)\ge \sum p_n\Phi(m_n)
	\ge \frac{\left(\sum p_nv_n\phi'(z_n)\right)^2}{\sum p_nv_n^2\phi''(z_n)}.
\end{equation*}
By \eqref{eq:Gconcave} and \eqref{eq:GsecondDerivative}, $G''(s)\le 0$ for every $s\in D(x,v)$. Hence $G(\cdot;p,x,v)$ is concave on $D(x,v)$. To verify the definition of concavity of $\cA_p$, take arbitrary $x^1,x^2\in C^N$ and $\alpha\in[0,1]$. Because $C^N$ is convex, definition~\eqref{eq:oneDimensionalDomain} gives $[0,1]\subset D(x^1,x^2-x^1)$. The concavity of the corresponding one-dimensional restriction therefore implies
\begin{align*}
	\cA_p((1-\alpha)x^1+\alpha x^2)
	&=G(\alpha;p,x^1,x^2-x^1)\\
	&\ge (1-\alpha)G(0;p,x^1,x^2-x^1)
	+\alpha G(1;p,x^1,x^2-x^1)\\
	&=(1-\alpha)\cA_p(x^1)+\alpha\cA_p(x^2).
\end{align*}
Thus $\cA_p$ is concave on $C^N$, so \ref{item:Phi2} holds.
\end{proof}

Section 3.4 of \citet{SuenConcave} derives the equivalence between the concavity of $\Phi$ and that of inverse absolute prudence, $-u''/u'''$, by differentiating twice. The next lemma establishes the same equivalence under the maintained assumption that $u$ is three times continuously differentiable, using local absolute continuity and almost-everywhere derivatives instead of imposing further smoothness.

\begin{lem}[Equivalent curvature conditions]\label{lem:PhiRatio}
Under Assumption~\ref{asmp:utility}, let $\Phi$ be as in Lemma~\ref{lem:Phi}. Then $\Phi$ is concave on $(0,\infty)$ if and only if inverse absolute prudence $\psi\coloneqq -u''/u'''$ is concave on $C$.
\end{lem}

\begin{proof}
Since $\phi=u'$, we have $\psi=-\phi'/\phi''$, so the definition of $\Phi$ gives
\begin{equation}
	\Phi(\phi(x))=-\phi'(x)\psi(x). \label{eq:PhiPsi}
\end{equation}
Suppose either $\Phi$ or $\psi$ is concave. Then that function is locally Lipschitz on the interior of its domain. Since $\phi:C\to(0,\infty)$ is a continuously differentiable bijection with $\phi'<0$, \eqref{eq:PhiPsi} implies that the other function is locally Lipschitz as well. Therefore both functions are locally absolutely continuous. Differentiating \eqref{eq:PhiPsi} at points of differentiability and using $\phi''\psi=-\phi'$, we obtain
\begin{equation*}
	\Phi'(\phi(x))=1-\psi'(x)
\end{equation*}
for almost every $x\in C$. Because $\phi$ is strictly decreasing, $\Phi'$ is decreasing almost everywhere on $(0,\infty)$ if and only if $\psi'$ is decreasing almost everywhere on $C$. The result now follows from the characterization of concavity of locally absolutely continuous functions by the monotonicity of their derivatives.
\end{proof}

The probability-measure case of the next lemma follows from a known result about expected-utility certainty equivalents. For the increasing, concave utility index $v=-\phi$, \citet[Theorem~1]{BenTalTeboulle1986} show that the certainty-equivalent functional generated by $v$ is concave if and only if the absolute risk tolerance of $v$ is concave. Because $-v'/v''=-u''/u'''$, their characterization yields the probability-measure case of the concavity assertion below.

The application here requires more than the probability-measure case. To pass from finite-support shocks to general conditional distributions, \citet{SuenConcave} approximates by finite partitions and takes pointwise limits. The lemma below combines this approximation idea with the certainty-equivalent result and extends it to finite measures and measurable functions. For a general utility function satisfying the condition in Theorem~\ref{thm:main}, the measure may be any subprobability measure; allowing mass below one is essential because the relevant conditional measure has mass $\E_t[\beta_{t+1}R_{t+1}]$, which can be strictly less than one. For HARA utility, the conclusion holds for a finite measure of arbitrary mass, as required by Theorem~\ref{thm:unrestricted}. The measure-theoretic form also builds on the concavity-preservation argument in condition (17), Remark 2.1, and Lemma B.8 of \citet{MaStachurskiToda2020JET}.

\begin{lem}[Integral aggregator]\label{lem:multivariateAggregator}
Under Assumption~\ref{asmp:utility}, let $\phi$ and $\Phi$ be as in Lemma~\ref{lem:Phi}. Let $(\Xi,\cX,\nu)$ be a finite measure space such that $0<\nu(\Xi)<\infty$, and let $\cH$ be the set of measurable functions $h:\Xi\to C$ such that $\int_\Xi\phi(h)\diff\nu<\infty$. For $h\in \cH$, define
\begin{equation*}
	\cA_\nu(h)\coloneqq\phi^{-1}\left(\int_\Xi\phi(h)\diff\nu\right).
\end{equation*}
Suppose either that $\nu(\Xi)\le 1$ and $\Phi$ is concave, or that $\Phi(m)=\kappa m$ for all $m>0$ and some $\kappa>0$. Then $\cH$ is convex and $\cA_\nu$ is increasing and concave on $\cH$. Consequently, if $S\subset\R$ is an interval and $h:S\times\Xi\to C$ is such that $h(s,\cdot)\in\cH$ for every $s\in S$ and $h(\cdot,\xi)$ is increasing and concave for every $\xi\in\Xi$, then
\begin{equation}
	s\mapsto\cA_\nu(h(s,\cdot)) \label{eq:multivariateAggregatorComposition}
\end{equation}
is increasing and concave on $S$.
\end{lem}

\begin{proof}
Since $C$ is convex and $\phi$ is convex, for $h_1,h_2\in\cH$ and $\alpha\in[0,1]$, the function $h_\alpha\coloneqq(1-\alpha)h_1+\alpha h_2$ maps $\Xi$ into $C$ and satisfies
\begin{equation*}
	\phi(h_\alpha)\le(1-\alpha)\phi(h_1)+\alpha\phi(h_2).
\end{equation*}
Hence $\int_\Xi\phi(h_\alpha)\diff\nu<\infty$, so $\cH$ is convex. If $h_1,h_2\in\cH$ and $h_1\le h_2$, then $\phi(h_1)\ge\phi(h_2)$ because $\phi$ is strictly decreasing. Since $\phi^{-1}$ is also strictly decreasing, it follows that $\cA_\nu(h_1)\le\cA_\nu(h_2)$. Thus $\cA_\nu$ is increasing.

To show the concavity of $\cA_\nu$, suppose first that $h_1$ and $h_2$ are simple functions on a common finite measurable partition $\set{E_n}_{n=1}^N$ of $\Xi$. Write $h_i=x_{in}$ on $E_n$ and $p_n\coloneqq\nu(E_n)$, omitting sets of zero measure. If $\nu(\Xi)\le 1$ and $\Phi$ is concave, Lemma~\ref{lem:Phi}\ref{item:Phi2} implies that $\cA_p$ is concave. Alternatively, suppose $\Phi(m)=\kappa m$. Then for arbitrary $p,m\in\R_{++}^N$,
\begin{equation*}
	\Phi\left(\sum p_nm_n\right)=\sum p_n\Phi(m_n).
\end{equation*}
The calculations leading to \eqref{eq:Gconcave} and \eqref{eq:CS}, which do not use the restriction on $\sum p_n$, together with the one-dimensional restriction argument in the proof of Lemma~\ref{lem:Phi}, show that $\cA_p$ is concave for arbitrary positive weights $p$. Thus, in either case, for every $\alpha\in[0,1]$,
\begin{align}
	\cA_\nu((1-\alpha)h_1+\alpha h_2)
	&=\cA_p((1-\alpha)x_1+\alpha x_2)\notag\\
	&\ge(1-\alpha)\cA_p(x_1)+\alpha\cA_p(x_2)\notag\\
	&=(1-\alpha)\cA_\nu(h_1)+\alpha\cA_\nu(h_2). \label{eq:simpleAggregatorConcavity}
\end{align}
For general $h_1,h_2\in\cH$, choose compact intervals $[a_k,b_k]\uparrow C$ and define $T_k(c)\coloneqq\min\set{\max\set{c,a_k},b_k}$. For each $k$, approximate $T_k(h_1)$ and $T_k(h_2)$ uniformly by simple functions taking values in $[a_k,b_k]$ and constant on the cells of a common finite measurable partition. The continuity of $\phi$ and $\phi^{-1}$ then allows passage to the uniform limit in \eqref{eq:simpleAggregatorConcavity}. Thus \eqref{eq:simpleAggregatorConcavity} holds with $h_i$ replaced by $T_k(h_i)$.

It remains to let $k\to\infty$. We have $T_k(h_i)\to h_i$ pointwise and
\begin{equation*}
	\phi(T_k(h_i))\le\phi(h_i)+\phi(b_1),\qquad i=1,2.
\end{equation*}
By convexity of $\phi$, these bounds also provide an integrable majorant for
\begin{equation*}
	\phi((1-\alpha)T_k(h_1)+\alpha T_k(h_2)).
\end{equation*}
The dominated convergence theorem therefore yields convergence of all three integrals in the truncated version of \eqref{eq:simpleAggregatorConcavity}. Using the continuity of $\phi^{-1}$ once more and letting $k\to\infty$ proves that $\cA_\nu$ is concave on $\cH$.

The monotonicity of the function in \eqref{eq:multivariateAggregatorComposition} follows directly from the monotonicity of $h(\cdot,\xi)$ and $\cA_\nu$. For concavity, fix $s_1,s_2\in S$ and $\alpha\in[0,1]$, and let $s_\alpha\coloneqq(1-\alpha)s_1+\alpha s_2$. Pointwise concavity of $h$ gives
\begin{equation*}
	h(s_\alpha,\cdot)\ge(1-\alpha)h(s_1,\cdot)+\alpha h(s_2,\cdot).
\end{equation*}
Applying first the monotonicity and then the concavity of $\cA_\nu$ proves that the function in \eqref{eq:multivariateAggregatorComposition} is concave.
\end{proof}

\subsection{Main characterization theorems}\label{subsec:proofMain}

The next proposition proves the necessity part of Theorem~\ref{thm:main} using only one-period problems with finite-support shocks and no borrowing constraint.

\begin{prop}[Necessity]\label{prop:necessity}
Under Assumption~\ref{asmp:utility}, suppose the consumption function is concave in every optimal saving problem in
\begin{equation*}
	\OS(\tilde{\beta},\tilde{R},\tilde{Y},\infty;\ \E[\beta R]\le1,\ T=1)
\end{equation*}
in which $(\beta,R,Y)$ has finite support. Then $-u''/u'''$ is concave.
\end{prop}

\begin{proof}
Fix $N\in\N$, $p,v\in\R_{++}^N$, and $x\in C^N$ such that $\lambda\coloneqq\sum_{n=1}^Np_n\le 1$. Choose $s_0\in\R$ sufficiently negative that $y_n\coloneqq x_n-v_ns_0>0$ for all $n$. Consider a one-period optimal saving problem in which $(\beta,R,Y)$ takes the value
\begin{equation}
	(\beta_n,R_n,y_n)\coloneqq\left(\frac{\lambda}{v_n},v_n,x_n-v_ns_0\right) \label{eq:construction}
\end{equation}
with probability $\pi_n\coloneqq p_n/\lambda$. This construction is admissible because $\E[\beta R]=\sum_{n=1}^N\pi_n\lambda=\lambda\le 1$. The function $g$ in \eqref{eq:g} associated with this problem satisfies
\begin{align}
	g(s_0+s)&=\phi^{-1}\left(\sum_{n=1}^N\pi_n\beta_nR_n\phi(R_n(s_0+s)+y_n)\right)\notag\\
	&=\phi^{-1}\left(\sum_{n=1}^Np_n\phi(x_n+v_ns)\right)=G(s;p,x,v). \label{eq:gG}
\end{align}
By assumption, the consumption function of this one-period problem is concave. Lemma~\ref{lem:g} therefore implies that $g$ is concave on its domain. It follows from \eqref{eq:gG} that $G(\cdot;p,x,v)$ is concave on the domain $D(x,v)$ defined in \eqref{eq:oneDimensionalDomain}. Since $N,p,x,v$ are arbitrary, Lemma~\ref{lem:Phi} implies that $\Phi$ is concave. Lemma~\ref{lem:PhiRatio} then implies that $-u''/u'''$ is concave.
\end{proof}

I next prove sufficiency for both characterization theorems. The proof follows the backward-induction argument of \citet{SuenConcave}: rewrite the Euler equation in terms of a marginal-utility aggregator and use its concavity to propagate concavity of the consumption policy. Lemma~\ref{lem:multivariateAggregator} permits this argument with stochastic discount factors, returns, income, and borrowing limits and with general conditional distributions. The notation $g_{t+1}$ below follows \eqref{eq:g} and \citet{Toda2021JME}: it gives the current consumption level implied by a choice of current saving, now accounting for the continuation consumption function.

\begin{prop}[Sufficiency]\label{prop:sufficiency}
Under Assumption~\ref{asmp:utility}, the following statements hold.
\begin{enumerate}
	\item\label{item:suffCI} If $-u''/u'''$ is concave, then every optimal saving problem in the conditionally impatient class \eqref{eq:conditionallyImpatientClass} has increasing and concave consumption functions and increasing and convex saving functions.
	\item\label{item:suffHARA} If $u$ is HARA, then the same conclusions hold for every optimal saving problem in $\OS(\tilde{\beta},\tilde{R},\tilde{Y},\tilde{b})$.
\end{enumerate}
\end{prop}

\begin{proof}
For claim~\ref{item:suffCI}, Lemma~\ref{lem:PhiRatio} implies that $\Phi$ is concave. For claim~\ref{item:suffHARA}, equations \eqref{eq:ARA}, \eqref{eq:PhiPsi}, and the expression for inverse absolute prudence following \eqref{eq:HARA} imply
\begin{equation}
	\Phi(m)=\frac{m}{a+1},\qquad m>0, \label{eq:PhiHARA}
\end{equation}
where $a+1>0$ because $u'''>0$. I prove both claims by backward induction. Since $b_T=0$ and utility is strictly increasing, the terminal policies on the feasible wealth domain are
\begin{equation*}
	c_T(w)=w
	\quad\text{and}\quad
	s_T(w)\coloneqq w-c_T(w)=0.
\end{equation*}
Thus $c_T$ is increasing and concave, while $s_T$ is increasing and convex.

Suppose the claim holds at time $t+1$, and fix a time $t$ history. On the conditional next-state space, define the finite measure $\nu_t(A)\coloneqq\E_t[\beta_{t+1}R_{t+1}\boldsymbol{1}_A]$. Then $0<\nu_t(\Omega)=\E_t[\beta_{t+1}R_{t+1}]<\infty$. In claim~\ref{item:suffCI}, conditional impatience \eqref{eq:conditionalImpatience} further implies $\nu_t(\Omega)\le 1$; in claim~\ref{item:suffHARA}, \eqref{eq:PhiHARA} applies. Hence Lemma~\ref{lem:multivariateAggregator} applies in either case.
For every saving level $s$ for which the expression is well-defined, let
\begin{align*}
	g_{t+1}(s)
	&\coloneqq\phi^{-1}\left(\E_t\left[\beta_{t+1}R_{t+1}
	\phi\left(c_{t+1}(R_{t+1}s+Y_{t+1})\right)\right]\right)\\
	&=\cA_{\nu_t}\left(c_{t+1}(R_{t+1}s+Y_{t+1})\right),
\end{align*}
where I suppress the state argument $\omega$ of $c_{t+1}$. By the induction hypothesis and $R_{t+1}>0$, the function $s\mapsto c_{t+1}(R_{t+1}s+Y_{t+1})$ is increasing and concave in every next-period state. Lemma~\ref{lem:multivariateAggregator} therefore implies that $g_{t+1}$ is increasing and concave. Define
\begin{equation*}
	W_{t+1}(s)\coloneqq s+g_{t+1}(s).
\end{equation*}
The function $W_{t+1}$ is strictly increasing and concave. Let $s_t^\circ\coloneqq W_{t+1}^{-1}$ on its range, which is the unconstrained saving function. The inverse of a strictly increasing concave function is increasing and convex. Moreover, if $w_1<w_2$, then
\begin{equation}
	0\le s_t^\circ(w_2)-s_t^\circ(w_1)\le w_2-w_1, \label{eq:savingLipschitz}
\end{equation}
because $W_{t+1}(s_2)-W_{t+1}(s_1)\ge s_2-s_1$ whenever $s_1<s_2$.

Suppose first that $b_t<\infty$. Feasibility of $w$ implies $w+b_t>\inf C$. If $w+b_t\in C$, the Euler inequality for current consumption can be written as
\begin{equation*}
	\phi(c_t(w))=\max\set{\E_t\left[\beta_{t+1}R_{t+1}
	\phi\left(c_{t+1}(R_{t+1}(w-c_t(w))+Y_{t+1})\right)\right],\phi(w+b_t)}.
\end{equation*}
An interior solution satisfies $c_t(w)=g_{t+1}(s_t(w))$ and $w=W_{t+1}(s_t(w))$, whereas a binding constraint gives $s_t(w)=-b_t$. Since $W_{t+1}$ is strictly increasing, the constraint binds exactly when $s_t^\circ(w)\le-b_t$.

If $w+b_t\ge\sup C$, the borrowing constraint is necessarily slack. Indeed, $w-s_t^\circ(w)=g_{t+1}(s_t^\circ(w))\in C$, so $s_t^\circ(w)>w-\sup C\ge-b_t$. If borrowing is unrestricted, the constraint is absent and $s_t=s_t^\circ$. Thus, in all cases,
\begin{equation*}
	s_t(w)=\max\set{-b_t,s_t^\circ(w)}.
\end{equation*}
Here the right-hand side is understood as $s_t^\circ(w)$ when $b_t=\infty$. The maximum of a constant and an increasing convex function is increasing and convex. It also preserves the upper bound in \eqref{eq:savingLipschitz}. Therefore $s_t$ is increasing, convex, and satisfies
\begin{equation*}
	0\le s_t(w_2)-s_t(w_1)\le w_2-w_1
\end{equation*}
whenever $w_1<w_2$. It follows that $c_t(w)=w-s_t(w)$ is increasing and concave. This completes the induction and proves both claims.
\end{proof}

\begin{proof}[Proof of Theorem~\ref{thm:unrestricted}]
If $u$ is HARA, Proposition~\ref{prop:sufficiency}\ref{item:suffHARA} implies that all consumption functions are concave. The converse is essentially contained in the proof of Lemma~3 of \citet{Toda2021JME}; I give a shorter argument using only the case $N=1$. Suppose consumption functions are concave in every problem in $\OS(\tilde{\beta},\tilde{R},\tilde{Y},\tilde{b})$. Let $\phi=u'$ and $\Phi$ be as in Lemma~\ref{lem:Phi}. Fix $\beta,m>0$, let $x\coloneqq\phi^{-1}(m)$, and choose $s_0<x$. Consider the deterministic one-period problem with unrestricted borrowing, discount factor $\beta$, and set $(R,Y)=(1,x-s_0)$. For the function $g$ in \eqref{eq:g}, we have
\begin{equation*}
	g(s_0+s)=\phi^{-1}\left(\beta\phi(x+s)\right).
\end{equation*}
By assumption, the consumption function is concave, so Lemma~\ref{lem:g} implies that the right-hand side is concave in a neighborhood of $s=0$. The calculation leading to \eqref{eq:Gconcave}, which does not use the restriction on the sum of the weights, applied at $s=0$ with $N=1$ and weight $\beta$ gives
\begin{equation}
	\Phi(\beta m)\ge \beta\Phi(m). \label{eq:PhiHomogeneousInequality}
\end{equation}
Since $\beta,m>0$ are arbitrary, applying \eqref{eq:PhiHomogeneousInequality} with $(\beta,m)$ replaced by $(1/\beta,\beta m)$ yields the reverse inequality. Therefore $\Phi(\beta m)=\beta\Phi(m)$ for all $\beta,m>0$, and hence $\Phi(m)=\kappa m$, where $\kappa\coloneqq\Phi(1)>0$.

Using absolute risk tolerance $\tau=-\phi/\phi'$ from \eqref{eq:riskTolerance}, the definition of $\Phi$ and the identity $\Phi(\phi(x))=\kappa\phi(x)$ imply
\begin{equation*}
	\tau'(x)=-1+\frac{\phi(x)\phi''(x)}{\phi'(x)^2}
	=-1+\frac{1}{\kappa}\eqqcolon a.
\end{equation*}
Thus $\tau(x)=ax+b$ for some $b\in\R$, so
\begin{equation*}
	-\frac{u''(x)}{u'(x)}=\frac{1}{ax+b}.
\end{equation*}
Solving $\phi'/\phi=-1/(ax+b)$ and applying the Inada conditions identifies $C$ with $\set{x\in\R:ax+b>0}$. Hence $u$ is HARA.
\end{proof}

\begin{proof}[Proof of Theorem~\ref{thm:main}]
If $-u''/u'''$ is concave, Proposition~\ref{prop:sufficiency}\ref{item:suffCI} implies that all consumption functions are concave. Conversely, suppose consumption functions are concave in every problem in the conditionally impatient class \eqref{eq:conditionallyImpatientClass}. This class includes the finite-support, one-period problems with unrestricted borrowing considered in Proposition~\ref{prop:necessity}. Hence $-u''/u'''$ is concave.
\end{proof}

\subsection{Perfect foresight and pathwise restrictions}\label{subsec:proofComparisons}

\begin{proof}[Proof of Proposition~\ref{prop:deterministic}]
Let $\phi=u'$. Under Assumption~\ref{asmp:utility}, $\phi$ is a twice continuously differentiable bijection from $C$ onto $(0,\infty)$ with $\phi'<0$, and $\tau=-\phi/\phi'$ is continuously differentiable.

First consider sufficiency. Fix a deterministic problem and a date $t<T$. For $j=t+1,\dots,T$, define the cumulative gross return and discounted gross return by
\begin{equation*}
	R_{t\to j}\coloneqq\prod_{k=t+1}^jR_k,
	\qquad
	\Lambda_{t\to j}\coloneqq\prod_{k=t+1}^j\beta_kR_k.
\end{equation*}
Iterating the Euler equation shows that future consumption, as a function of current consumption $x$, is
\begin{equation}
	h_{t,j}(x)\coloneqq\phi^{-1}\left(\frac{\phi(x)}{\Lambda_{t\to j}}\right). \label{eq:deterministicFutureConsumption}
\end{equation}
Dividing the derivative of $\phi(h_{t,j}(x))=\phi(x)/\Lambda_{t\to j}$ by this identity and using $\phi'/\phi=-1/\tau$ gives
\begin{equation*}
	h_{t,j}'(x)=\frac{\tau(h_{t,j}(x))}{\tau(x)}>0.
\end{equation*}
Differentiating once more yields
\begin{equation}
	h_{t,j}''(x)=\frac{\tau(h_{t,j}(x))}{\tau(x)^2}(\tau'(h_{t,j}(x))-\tau'(x)). \label{eq:deterministicFutureCurvature}
\end{equation}
This calculation uses only the smoothness in Assumption~\ref{asmp:utility}.

Since borrowing is unrestricted and $c_T=w_T$, the present-value budget constraint is
\begin{equation}
	w+\sum_{j=t+1}^T\frac{Y_j}{R_{t\to j}}
	=x+\sum_{j=t+1}^T\frac{h_{t,j}(x)}{R_{t\to j}}
	\eqqcolon H_t(x). \label{eq:deterministicPresentValue}
\end{equation}
The function $H_t$ is strictly increasing. The Inada conditions imply that each $h_{t,j}$ approaches the same endpoints of $C$ as $x$, so the range of $H_t$ is exactly the feasible interval of present-value resources. The consumption plan specified by \eqref{eq:deterministicFutureConsumption} and \eqref{eq:deterministicPresentValue} satisfies the first-order conditions and is the unique optimum by strict concavity. Therefore,
\begin{equation*}
	c_t(w)=H_t^{-1}\left(w+\sum_{j=t+1}^T\frac{Y_j}{R_{t\to j}}\right).
\end{equation*}
Since $H_t(c_t(w))=w+\sum_{j=t+1}^TY_j/R_{t\to j}$, differentiating with respect to $w$ gives
\begin{align*}
	H_t'(c_t(w))c_t'(w)&=1,\\
	H_t''(c_t(w))(c_t'(w))^2+H_t'(c_t(w))c_t''(w)&=0.
\end{align*}
Solving for $c_t''(w)$ and setting $x=c_t(w)$ yields
\begin{equation}
	c_t''(w)=-\frac{H_t''(x)}{H_t'(x)^3}. \label{eq:deterministicPolicyCurvature}
\end{equation}

If $\beta_kR_k\le1$ for every transition, then $\Lambda_{t\to j}\le1$ and $h_{t,j}(x)\le x$ by \eqref{eq:deterministicFutureConsumption}. For concave $\tau$, its derivative is decreasing, so \eqref{eq:deterministicFutureCurvature} implies $h_{t,j}''(x)\ge0$. Hence $H_t''(x)\ge0$ and $c_t''(w)\le0$. If instead $\beta_kR_k\ge1$ for every transition and $\tau$ is convex, then $h_{t,j}(x)\ge x$ and $\tau'$ is increasing, giving the same conclusion. Reversing the curvature of $\tau$ in either case gives $h_{t,j}''(x)\le0$ and therefore $c_t''(w)\ge0$. The terminal consumption function $c_T(w)=w$ is affine.

If $\tau$ has the corresponding strict curvature and every transition satisfies the relevant strict return inequality, then $h_{t,j}(x)\ne x$ and the comparison of derivatives in \eqref{eq:deterministicFutureCurvature} is strict. Equation \eqref{eq:deterministicPolicyCurvature} therefore gives strict consumption curvature for $t<T$. If $\tau$ is affine, all second derivatives in \eqref{eq:deterministicFutureCurvature} vanish, so each $c_t$ is affine. If every $\beta_kR_k=1$, then $h_{t,j}(x)=x$, which gives the same conclusion without a restriction on the curvature of $\tau$.

For necessity, it suffices to consider one-period problems with $R=1$ and $Y=0$. Take any $x,y\in C$ and set the initial wealth $w=x+y$ and discount factor $\beta=\phi(x)/\phi(y)$. Then $(c_0,c_1)=(x,y)$ satisfies the budget constraint and the Euler equation, and hence is the unique optimum. Fixing $\beta$, consider perturbing $w$. The Euler equation is then
\begin{equation*}
	\phi(x)-\beta\phi(w-x)=0.
\end{equation*}
Applying the implicit function theorem, we obtain
\begin{equation*}
	c_0'(w)=\frac{\partial x}{\partial w}=\frac{\beta\phi'(w-x)}{\phi'(x)+\beta\phi'(w-x)}=\frac{\tau(x)}{\tau(x)+\tau(y)},
\end{equation*}
where we have used $\tau=-\phi/\phi'$ and $y=w-x$. Differentiating once again, we obtain
\begin{equation*}
	c_0''(w)=\frac{\tau(x)\tau(y)}{(\tau(x)+\tau(y))^3}
		(\tau'(x)-\tau'(y)).
\end{equation*}
If $x>y$, then $\beta<1$, so uniform consumption concavity in part~\ref{item:deterministicImpatient} implies $\tau'(x)\le\tau'(y)$ for every such pair. Thus $\tau'$ is decreasing and $\tau$ is concave. If $x<y$, then $\beta>1$, and the same inequality implied by part~\ref{item:deterministicPatient} instead shows that $\tau'$ is increasing and $\tau$ is convex. Uniform consumption convexity reverses both inequalities and yields the opposite curvature conditions.
\end{proof}

\begin{proof}[Proof of Corollary~\ref{cor:deterministicAffine}]
If the consumption functions are affine in every problem in $\OS(\beta,R,Y,\infty)$, they are both concave and convex in every such problem satisfying $\beta_tR_t\le1$ at every transition. Proposition~\ref{prop:deterministic}\ref{item:deterministicImpatient} therefore implies that absolute risk tolerance $\tau$ is both concave and convex, hence affine. Under Assumption~\ref{asmp:utility}, this is equivalent to $u$ being HARA.

Conversely, if $u$ is HARA, then $\tau'$ is constant, so \eqref{eq:deterministicFutureCurvature} gives $h_{t,j}''=0$ for arbitrary paths of discounted returns. Equations~\eqref{eq:deterministicPresentValue} and \eqref{eq:deterministicPolicyCurvature} imply that $c_t''=0$ for every $t<T$. The terminal consumption function $c_T(w)=w$ is also affine.
\end{proof}

\begin{proof}[Proof of Proposition~\ref{prop:reverseImpatience}]
I first establish a consequence common to both parts. Let $\phi=u'$ and $\Phi$ be as in Lemma~\ref{lem:Phi}. Both hypotheses include every finite-support one-period problem with unrestricted borrowing and $\beta R=1$ almost surely. In the construction \eqref{eq:construction} with $\lambda=1$, we have $\beta_n=1/v_n$ and $R_n=v_n$, so the product equals one in every state. The proof of Proposition~\ref{prop:necessity} therefore gives concavity of $G(\cdot;p,x,v)$ whenever $\sum p_n=1$. The equality-case argument \eqref{eq:CSequality} in the proof of Lemma~\ref{lem:Phi} yields
\begin{equation*}
	\Phi\left(\sum_{n=1}^N p_nm_n\right)\ge\sum_{n=1}^N p_n\Phi(m_n)
\end{equation*}
for every $N\in\N$ and $p,m\in\R_{++}^N$ with $\sum p_n=1$. Hence $\Phi$ is concave.

For part~\ref{item:pathwisePatient}, the deterministic one-period construction in the proof of Theorem~\ref{thm:unrestricted} has $R=1$, so it satisfies the required pathwise restriction whenever $\beta\ge1$. The same calculation therefore gives
\begin{equation*}
	\Phi(\beta m)\ge\beta\Phi(m),\qquad \beta\ge1,\quad m>0.
\end{equation*}
For the reverse inequality, fix $\beta>1$ and $m>0$. Concavity and positivity of $\Phi$ imply, for every $\epsilon>0$,
\begin{equation*}
	\Phi\left(m+(1-\beta^{-1})\epsilon\right)
	\ge\beta^{-1}\Phi(\beta m)+(1-\beta^{-1})\Phi(\epsilon)
	\ge\beta^{-1}\Phi(\beta m).
\end{equation*}
Letting $\epsilon\downarrow0$ and using continuity gives $\Phi(\beta m)\le\beta\Phi(m)$. Equality is immediate when $\beta=1$. Thus $\Phi(\beta m)=\beta\Phi(m)$ for all $\beta\ge1$ and $m>0$, which implies $\Phi(m)=\kappa m$ with $\kappa=\Phi(1)>0$. The final argument in the proof of Theorem~\ref{thm:unrestricted} then implies that $u$ is HARA.

For part~\ref{item:pathwiseImpatient}, the concavity of $\Phi$ established at the start of the proof and Lemma~\ref{lem:PhiRatio} imply that $\psi$ is concave.
\end{proof}

\section*{Declaration of generative AI and AI-assisted technologies in the manuscript preparation process}

During the preparation of this work, the author used ChatGPT (OpenAI) to assist with language editing, manuscript organization, literature searches, and checks of exposition and internal consistency. The author reviewed and edited all AI-assisted output and takes full responsibility for the content of the published article.

\printbibliography

\end{document}